\documentclass{article}
\usepackage{graphicx}
\usepackage{amsmath, amsthm, latexsym}
\usepackage{amssymb}

\newtheorem{theorem}{Theorem}
\newtheorem{lemma}{Lemma}
\newtheorem{definition}{Definition}
\newtheorem{corollary}{Corollary}

\newtheorem{remark}{Remark}
\usepackage{pdfsync}
      
\begin{document}

\vspace*{3cm} \thispagestyle{empty}
\vspace{5mm}

\noindent \textbf{\Large Fermi coordinates, simultaneity, and expanding space in Robertson-Walker cosmologies}\\

\textbf{\normalsize  \textbf{\normalsize David Klein}\footnote{Department of Mathematics and Interdisciplinary Research Institute for the Sciences, California State University, Northridge, Northridge, CA 91330-8313. Email: david.klein@csun.edu.}
and Evan Randles}\footnote{Department of Mathematics, California State University, Northridge, Northridge, CA 91330-8313. Email: evan.randles@my.csun.edu}\\

\vspace{4mm} \parbox{11cm}{\noindent{\small Explicit Fermi coordinates are given for geodesic observers comoving with the Hubble flow in expanding Robertson-Walker spacetimes, along with exact expressions for the metric tensors in Fermi coordinates. For the case of non inflationary cosmologies, it is shown that Fermi coordinate charts are global, and space-time is foliated by space slices of constant Fermi (proper) time that have finite extent.  A universal upper bound for the proper radius of any leaf of the foliation, i.e., for the proper radius of the spatial universe at any fixed time of the geodesic observer, is given.  A general expression is derived for the geometrically defined Fermi relative velocity of a test particle (e.g. a galaxy cluster) comoving with the Hubble flow away from the observer.  Least upper bounds of superluminal recessional Fermi velocities are given for spacetimes whose scale factors follow power laws, including matter-dominated and radiation-dominated  cosmologies. Exact expressions for the proper radius of any leaf of the foliation for this same class of spacetimes are given. It is shown that the radii increase linearly with proper time of the observer moving with the Hubble flow. These results are applied to particular cosmological models.}\\


}\\
\vspace{6cm}
\pagebreak

\setlength{\textwidth}{27pc}
\setlength{\textheight}{43pc}
\noindent \textbf{{\normalsize 1. Introduction}}\\

\noindent The usual foliation of a Robertson-Walker spacetime ($\mathcal{M}, g)$ by maximally symmetric space slices $\{\Sigma_{t}\}$, parameterized by synchronous proper time $t$, determines a notion of simultaneity and 
leads in a natural way to Hubble's law,

\begin{equation}\label{introhubble}
\dot{d}(t)\equiv v_{H}=Hd.
\end{equation}
Here $H$ is the Hubble parameter, $d$ is the proper distance on $\Sigma_{t}$ from the observer to the test particle (such as a galaxy cluster), and the overdot on $d$ signifies differentiation with respect to $t$.  If $H>0$ and the distance $d$ is sufficiently large, the Hubble speed $v_{H}$ exceeds the speed of light. It is argued largely on this basis that in a physical sense, the universe is expanding, c.f., for example, \cite{gron, CG, cook, confusion}. \\

\noindent However, for a given geodesic observer with world line $\gamma(t)$ following the Hubble flow, another geometrically natural foliation by space slices $\{\mathcal{M}_{\tau}\}$, previously considered in \cite{rindler, page, ellis}, is also available and leads to a different definition of simultaneity, and to a different notion of relative velocity.  To define $\mathcal{M}_{\tau}$, let $\varphi_{\tau} : \mathcal{M} \rightarrow \mathbb{R}$ by,

\begin{equation}\label{slice}
\varphi_{\tau}(p)=g(\exp_{\gamma(\tau)}^{-1}p,\, \gamma^{\prime}(\tau)).
\end{equation}
Then,
\begin{equation}\label{slice2}
\mathcal{M}_{\tau}\equiv \varphi_{\tau}^{-1}(0).
\end{equation}

\noindent In Eq.\eqref{slice} the exponential map, $\exp_{p}(v)$, denotes the evaluation at affine parameter $1$ of the geodesic starting at the point $p\in\mathcal{M}$, with initial derivative $v$.\\

\noindent  We refer to $\mathcal{M}_{\tau}$ as the \emph{Fermi space slice of $\tau$-simultaneous events} and to the observer following the timelike geodesic path $\gamma(t)$ as the \emph{Fermi observer}.  $\mathcal{M}_{\tau}$ consists of all the spacelike geodesics orthogonal to the path of the Fermi observer at fixed proper time $t=\tau$.  The restriction $g_{\tau}$ of the space-time metric $g$ to $\mathcal{M}_{\tau}$ makes $(\mathcal{M}_{\tau}, g_{\tau})$ a Riemannian manifold. \\

\noindent Fermi coordinates are associated to the foliation $\{\mathcal{M}_{\tau}\}$ in a natural way.  Each spacetime point on $\mathcal{M}_{\tau}$ is assigned time coordinate $\tau$, and the spatial coordinates are defined relative to a parallel transported orthonormal reference frame.   Specifically, a Fermi coordinate system \cite{walker, MTW, MM63, LN79} along $\gamma$ is determined by an orthonormal tetrad of vectors, $e_{0}(\tau), e_{1}(\tau), e_{2}(\tau), e_{3}(\tau)$ parallel along $\gamma$, where  $e_{0}(\tau)$ is the four-velocity of the Fermi observer, i.e., the unit tangent vector of $\gamma(\tau)$. Fermi coordinates $x^{0}$, $x^{1}$, $x^{2}$, $x^{3}$ relative to this tetrad   are defined by,

\begin{equation}\label{F2}
\begin{split}
x^{0}\left (\exp_{\gamma(\tau)} (\lambda^{j}e_{j}(\tau))\right)&= \tau \\
x^{k}\left (\exp_{\gamma(\tau)} (\lambda^{j}e_{j}(\tau))\right)&= \lambda^{k}, 
\end{split} 
\end{equation}

\noindent where Latin indices  run over $1,2,3$ (and Greek indices run over $0,1,2,3$).  Here it is assumed that the $\lambda^{j}$ are sufficiently small so that the exponential maps in Eq.\eqref{F2} are defined.\\

\noindent From the theory of differential equations, a solution to the geodesic equations depends smoothly on its initial data so it follows that Fermi coordinates are smooth, and it may be shown in general that there  exists a neighborhood $U$ of $\gamma$ on which the map $ (x^{0}, x^{1}, x^{2}, x^{3}): U \rightarrow \mathbb{R}^{4}$ is well-defined and is a diffeomorphism onto its image, \cite{oniell} (p. 200).   We refer to $(x^{\alpha}, U)$ as a Fermi coordinate chart for $\gamma$.  \\

\noindent A particularly useful feature of Fermi coordinates is that the metric tensor expressed in these coordinates is Minkowskian to first order near the path of the Fermi observer, with second order corrections involving only the curvature tensor \cite{MM63}. General formulas in the form of Taylor expansions for coordinate transformations to and from more general Fermi-Walker coordinates were given in \cite{KC1} and exact transformation formulas for a class of spacetimes were given in \cite{CM, KC3}. Applications of these coordinate systems are voluminous. They include the study of tidal dynamics, gravitational waves, statistical mechanics, and the influence of curved space-time on quantum mechanical phenomena \cite{CM, Ishii, Marz, FG, KC9, KY, B,P80, PP82}.\\

\noindent In this paper, we find explicit expressions for the Fermi coordinates $(x^{0}, x^{1}, x^{2}, x^{3})$ for Robertson-Walker cosmologies when the scale factor $a(t)$ (see Eq.\eqref{frwmetric}) is an increasing function of $t$.  We show  that the Fermi chart $(x^{\alpha}, U)$ for the Fermi observer $\gamma(t)$ in non inflationary\footnote{A Robertson-Walker space-time is non inflationary if $\ddot{a}(t)\leq0$ for all $t$.} Robertson-Walker space-times is global, i.e., we may take $U=\mathcal{M}$.\\  

\noindent We prove that the proper radius $\rho_{\mathcal{M}_{\tau}}$ of $\mathcal{M}_{\tau}$ is an increasing function of $\tau$ and is bounded above by the speed of light divided by the Hubble pararmeter at time $\tau$, i.e., $1/H$.  We show, in addition, that synchronous time $t$ decreases to zero along any spacelike geodesic, $Y_{\tau}(\rho)$, orthogonal to the path of the Fermi observer at fixed proper time $t=\tau$, as the proper distance $\rho\rightarrow\rho_{\mathcal{M}_{\tau}}$. Thus, if the cosmological model includes a big bang, i.e., satisfies the condition that $a(0^{+}) = 0$, then at any proper time $\tau$ of the Fermi observer, the geodesic path $Y_{\tau}(\rho)$ in $\mathcal{M}_{\tau}$ terminates at the big bang. In this sense, also noted in \cite{page, ellis}, the big bang is simulaneous with all spacetime events for this class of cosmological models.\\

\noindent As a special case, for models with scale factors of the form $a(t)=t^{\alpha}$ for some  $0< \alpha\leq1$, including matter-dominated and radiation-dominated cosmologies, we find an exact expression for the proper radius of the Fermi space slice of $\tau$-simultaneous events,

\begin{equation}\label{MradiusIntro}
\rho_{\mathcal{M}_{\tau}}=\tau \frac{\sqrt{\pi}\,\,\Gamma(\frac{1}{2\alpha}+\frac{1}{2})}{2\alpha\,\,\Gamma(\frac{1}{2\alpha}+1)}.
\end{equation}\\

\noindent General relativity restricts the speed of a test particle to be less than the speed of light relative to an observer at the exact spacetime point of the test particle, but for test particles and observers located at different space-time points, the theory provides no \textit{a priori} definition of relative velocity, and hence no upper bounds on speeds. Distant particles may have  superluminal or only sub light speeds, depending on the coordinate system used for the calculations, and on the particular definition of relative velocity. \\ 

\noindent Mitigating such ambiguities, four geometrically defined, non equivalent notions of relative velocity were introduced and developed in a series of papers by V. Bol\'os (c.f. \cite{bolos2, bolos} and the references therein).  Of these, the Fermi relative velocity is most closely analogous to the Hubble relative velocity of Eq.\eqref{introhubble}, and lends itself to comparison.  For the case of a test particle receding radially from the Fermi observer, the Fermi relative velocity may be described as follows. \\

\noindent The world line of a receding test particle intersects each space slice $\mathcal{M}_{\tau}$ at a point $Y_{\tau}(\rho(\tau))$ on the spacelike geodesic $Y_{\tau}(\rho)$ in $\mathcal{M}_{\tau}$.  The Fermi speed for such a particle is,
\begin{equation}\label{introvf}
v_{F}=\frac{d}{d\tau}\rho(\tau).
\end{equation}
In Eq.\eqref{introvf}, $\rho(\tau)$ is the proper distance at proper time $\tau$ from the Fermi observer to the test particle's position in $\mathcal{M}_{\tau}$. General definitions and properties of Fermi relative velocity for observers and test particles following arbitrary timelike paths are included in \cite{bolos}, and Fermi relative velocities were applied in \cite{KC10} to study the influence of a cosmological constant on receding test particles in the vacuum surrounding a galaxy supercluster.\footnote{We remark that Hubble velocity was incorrectly referred to as Fermi relative velocity in \cite{KC10}.} \\

\noindent In the present paper, we find a general expression for the Fermi speeds of test particles comoving with the Hubble flow relative to the Fermi observer. We prove that the Fermi relative velocities can exceed the speed of light in cosmological models with $a(t)=t^{\alpha}$ for some  $0< \alpha <1$ and find sharp upper bounds proportional to $\rho_{\mathcal{M}_{\tau}}$ of those speeds.  We show in particular that superluminal relative velocities occur in matter dominated and radiation dominated  cosmologies, but not in the de Sitter universe, in contrast to superluminal Hubble speeds.\\

\noindent The paper is organized as follows.  In Sect. 2 we give basic definitions and solve the geodesic equations for spacelike geodesics orthogonal to the Fermi observer's path, $\gamma(t)$.  In Sect. 3, we construct Fermi coordinates, express the metric tensor in these coordinates, and prove that the Fermi chart is global for non inflationary cosmologies.  Sect. 4 begins with the universal upper bound for $\rho_{\mathcal{M}_{\tau}}$, and then gives exact expressions for Fermi relative speeds of comoving test particles in expanding Robertson-Walker spacetimes.  Included are superluminal least upper bounds for the relative Fermi speeds and exact values for $\rho_{\mathcal{M}_{\tau}}$ when the scale factors follow power laws.  Sect 5 illustrates and applies the preceding results to particular cosmological models, and Sect. 6 is devoted to concluding remarks and discussion of the expansion of space.   \\

\noindent \textbf{{\normalsize 2. Definitions and spacelike geodesics}}\\

\noindent The Robertson-Walker metric on space-time $\mathcal{M}=\mathcal{M}_{k}$ is given by the line element,

\begin{equation}\label{frwmetric}
ds^2=-dt^2+a(t)^2\left(d\chi^2+S^2_k(\chi)d\Omega^2\right),
\end{equation}

\noindent where $d\Omega^2=d\theta^{2}+\sin^{2}\theta \,d\varphi^{2}$, $a(t)$ is the scale factor, and, 

\begin{equation}\label{Sk}
S_{k}(\chi)=
\begin{cases}
\sin(\chi) &\text{if}\,\, k= 1\\
 \chi&\text{if}\,\,k= 0\\ 
 \sinh(\chi)&\text{if}\,\,k=-1.
\end{cases}
\end{equation}

\noindent The coordinate $t>0$ is synchronous proper time and $\chi, \theta, \varphi$ are dimensionless. The values $1,0,-1$ of the parameter $k$ distinguish the three possible maximally symmetric space slices for constant values of $t$ with positive, zero, and negative curvatures respectively.\\

\noindent  There is a coordinate singularity in Eq.\eqref{frwmetric} at $\chi=0$, but this will not affect the calculations that follow.  Consider the submanifold $\mathcal{M}_{\theta_0,\varphi_0}=\mathcal{M}_{\theta_0,\varphi_0,k}$ determined by $\theta=\theta_{0}$ and $\varphi=\varphi_{0}$.  The restriction of the metric to $\mathcal{M}_{\theta_0,\varphi_0}$ is given by, 

\begin{equation}\label{frwmetric2}
ds^{2}=-dt^2+a(t)^{2} d\chi^{2}.
\end{equation}  

\noindent On $\mathcal{M}_{\theta_0,\varphi_0}$, the coordinate $\chi$ can be extended to take all real values  if $k=0$ or $-1$, and for the case that $k=1$, the range of $\chi$ is an interval centered at zero, so there is no coordinate singularity at $\chi=0$ on the submanifold (see e.g., \cite{GP}).\\ 

\noindent Consider the observer with timelike geodesic path, $\gamma(t)=(t,0)$ in $\mathcal{M}_{\theta_0,\varphi_0}$. Our immediate aim is to find expressions for all spacelike geodesics orthogonal to $\gamma(t)$. From the Lagrangian for Eq.\eqref{frwmetric2}, it follows that $a(t)^{2} d\chi/d\rho$ is a constant $C$ along geodesics parametrized by arc length $\rho$.  Since the tangent vector to the geodesic has unit length we get,

\begin{equation}\label{tdot}
(\dot{t})^{2}\equiv\left(\frac{dt}{d\rho}\right)^{2}= \frac{C^{2}}{a(t)^{2}}-1.
\end{equation}

\noindent The requirement that $X=(\dot{t},\dot{\chi})$ is orthogonal to $(1,0)$, the tangent vector to $\gamma(t)$ at $t=\tau$, forces $C = a(\tau)\equiv a_{0}$.  We will assume throughout that $a(t)$ is an increasing function of $t$, so that by Eq.\eqref{tdot}, $\dot{t}<0$, and therefore,

\begin{equation}\label{X}
X=-\sqrt{\left(\frac{a_0}{a(t)}\right)^2-1}\,\,\frac{\partial}{\partial t}+\frac{a_0}{a^2(t)}\,\frac{\partial}{\partial \chi}.
\end{equation}

\begin{remark}\label{submanifold}
The arc length parameter $\rho$ for the geodesic $Y(\rho)=(t(\rho),\chi(\rho))$ with tangent vector $X$ may be chosen so that $Y(0)=(t(0),\chi(0))=\gamma(\tau)$.  With this convention, which we assume throughout, it follows from symmetry that $\chi(\rho)$ is an odd function of $\rho$. 
\end{remark}

\noindent In light of Remark \ref{submanifold} there is no loss of generality in restricting our attention to those spacetime points with space coordinate $\chi\geq0$ corresponding to $\rho\geq0$ for the purpose of finding spacelike geodesics orthogonal to $\gamma(t)$ with initial point on $\gamma(t)$.

\begin{remark}\label{bolos}
A general expression for $X$ in Cartesian coordinates was given in \cite{bolos} and may be deduced from Eq.\eqref{X} using the transformation of space coordinates, $S_{k}(\chi)= x/(1+\frac{1}{4}kx^{2})$. 
\end{remark}

\noindent The vector field $X$ can be integrated to give explicit formulas for the geodesic, $Y(\rho)$,  for the special cases that $a(t) = \exp(H_{0}t)$, where $H_{0}$ is Hubble constant, i.e., for the de Sitter universe \cite{CM}, and for the Milne universe (see, e.g., \cite{cook}), i.e., for $a(t) = t$.  To obtain integral expressions for the general case, we introduce a change of parameter from $\rho$ to $\sigma$.  For the sake of clarity of exposition, we assume henceforth  that $k= 0$ or $-1$ so that the range of $\chi$ is unrestricted. The minor modifications needed to deal with the case $k=1$ would include altering the range of $\sigma$ indicated in Eq.\eqref{parameter}, but the methods are the same.  Let,

\begin{equation}\label{parameter}
\sigma=\left(\frac{a_0}{a(t)}\right)^2=a_0\dot{\chi}\quad\text{and}\quad \sigma\in \left[1,\sigma_{\infty}(\tau)\right),
\end{equation}
where 
\begin{equation}\label{range}
\sigma_{\infty}(\tau)\equiv
\begin{cases}
\left(a(\tau)/a_{\text {inf}}\right)^{2}&\text{if}\,\,a_{\text {inf}}\equiv \lim_{t\rightarrow0^{+}}a(t)>0\\
\infty&\text{if}\,\,\lim_{t\rightarrow0^{+}}a(t)=0.
\end{cases}
\end{equation}

\noindent Using Eq.\eqref{X}, it follows that,

\begin{equation}\label{tsigma}
\dot{t}=-\sqrt{\sigma-1},
\end{equation}

\noindent and then differentiating Eq.\eqref{parameter} gives,

\begin{equation}\label{sigmadot}
\frac{d\sigma}{d\rho}=a_0\ddot{\chi}=2\frac{\dot{a}(t)}{a(t)}\sigma \sqrt{\sigma-1},
\end{equation}

\noindent where $\dot{a}(t)$ denotes $da/dt$. From the chain rule,

\begin{equation}\label{dsigmadchi}
\frac{d\chi}{d\rho}=\frac{\sigma}{a_0}=\frac{d\chi}{d\sigma}\frac{d\sigma}{d\rho},
\end{equation}

\noindent and combining this with Eq.\eqref{sigmadot} gives,

\begin{equation}\label{dchidsigma}
\frac{d\chi}{d\sigma}=\frac{a(t)}{2a_0 \dot{a}(t)\sqrt{\sigma-1}}.
\end{equation}

\noindent We end this section with a theorem and corollary that give explicit integral formulas for spacelike geodesics orthogonal to the timelike path $\gamma(t)$.

\begin{theorem}\label{general} Let $a(t)$ be a smooth, increasing function of $t$ with inverse function $b(t)$.  Then the spacelike geodesic orthogonal to $\gamma(t)$ at $t=\tau$ and parametrized by the (non-affine) parameter $\sigma$ is given by $Y_{\tau}(\sigma)=(t(\tau,\sigma),\chi(\tau,\sigma))$ where,

\begin{eqnarray}
t(\tau,\sigma)&=&b\left(\frac{a(\tau)}{\sqrt{\sigma}}\right)\label{thm1}\\
\chi(\tau,\sigma)&=&\frac{1}{2}\int_{1}^{\sigma}\dot{b}\left(\frac{a(\tau)}{\sqrt{\tilde{\sigma}}}\right)\frac{1}{\sqrt{\tilde{\sigma}}\sqrt{\tilde{\sigma}-1}}d\tilde{\sigma},\label{thm2}
\end{eqnarray}
\noindent and where the overdot on $b$ denotes differentiation.  Moreover, for fixed $\tau$, the arc length $\rho$ along $Y_{\tau}(\sigma)$ is given by,

\begin{equation}\label{thm3}
\rho=\rho_{\tau}(\sigma)=\frac{a(\tau)}{2}\int_{1}^{\sigma}\dot{b}\left(\frac{a(\tau)}{\sqrt{\tilde{\sigma}}}\right)\frac{1}{\tilde{\sigma}^{3/2}\sqrt{\tilde{\sigma}-1}}d\tilde{\sigma}.
\end{equation}

\end{theorem}

\begin{proof}
Eq.\eqref{thm1} follows immediately from Eq.\eqref{parameter}.  To prove the other two equations, observe that by the inverse function theorem,

\begin{equation}\label{thm4}
\dot{b}(a(t))=\frac{1}{\dot{a}(t)}.
\end{equation}

\noindent From Eq.\eqref{parameter}, $a(t)=a_{0}/\sqrt{\sigma}$ and combining this with Eq.\eqref{thm4} gives,

\begin{equation}\label{goodlemma}
\frac{a(t)}{\dot{a}(t)}=\frac{a(\tau)}{\sqrt{\sigma}}\,\,\dot{b}\left(\frac{a(\tau)}{\sqrt{\sigma}}\right).
\end{equation}

\noindent Substituting Eq.\eqref{goodlemma} into Eqs.\eqref{dchidsigma} and \eqref{sigmadot} and integrating yields Eqs.\eqref{thm2} and \eqref{thm3}.
\end{proof}

\noindent From Eq.\eqref{thm3} it follows that for a fixed value of $\tau$,  $\rho$ is a smooth, increasing function of  $\sigma$ with a smooth inverse which we denote by,

\begin{equation}\label{inverse}
\sigma_{\tau}(\rho)=\sigma(\rho).
\end{equation}
Combining Eq.\eqref{inverse} with Theorem \ref{general} immediately gives the following corollary.

\begin{corollary}\label{affinegeodesic}
 Let $a(t)$ be a smooth, increasing function of $t$ with inverse function $b(t)$.  Then the spacelike geodesic orthogonal to $\gamma(t)$ at $t=\tau$, and parametrized by arc length $\rho$, is given by

\begin{equation}\label{Y(rho)}
Y_{\tau}(\rho)=(t(\tau,\sigma(\rho)),\chi(\tau,\sigma(\rho))).
\end{equation}

\end{corollary}

\noindent \textbf{{\normalsize 3. Fermi coordinates}}\\

\noindent In this section we find explicit Fermi coordinates for the timelike geodesic co-moving observer in Robertson-Walker cosmologies, and show that with suitable assumptions on the scale factor $a(t)$, Fermi coordinates cover the entire spacetime. We begin with some technical results needed for that purpose.\\

\noindent In this section, we use the following notation,

\begin{equation}\label{openU}
 U=\{(\tau,\sigma): \tau>0\, \text{and}\,\sigma\in \left(1,\sigma_{\infty}(\tau)\right)\}
\end{equation}
and
\begin{equation}\label{notopenU}
 U_1=\{(\tau,\sigma): \tau>0\, \text{and}\,\sigma\in \left[1,\sigma_{\infty}(\tau)\right)\},
\end{equation}
where $\sigma_{\infty}(\tau)$ is given by Eq.\eqref{range}.  Observe that $U$ is an open subset of $\mathbb{R}^{2}$.

\begin{lemma}\label{lem1}
In addition to the hypotheses of Theorem \ref{general}, assume that $a(t)$ is unbounded and $\ddot{b}(t)\geq0$ for all $t>0$. Then the map $F :U_{1}\to (0,\infty)\times[0,\infty)$ given by,

\begin{equation}\label{Function}
F(\tau,\sigma)=\left(t(\tau,\sigma), \chi(\tau,\sigma)\right)=Y_{\tau}(\sigma),
\end{equation}

\noindent is a bijection, and $F :U\to (0,\infty)\times(0,\infty)$ is a diffeomorphism. Here, the functions $t$ and $\chi$ are defined by Eqs.\eqref{thm1} and \eqref{thm2} respectively.
\end{lemma}

\begin{proof}

Let $(t_1,\chi_1)\in(0,\infty)\times[0,\infty)$ be arbitrary but fixed. We show that $F(\tau_{1},\sigma_{1})=(t_1,\chi_1)$ for a uniquely determined pair $(\tau_{1},\sigma_{1})\in U_{1}$. From Eq.\eqref{thm1} we must have,

\begin{equation}\label{sigma1}
\sigma_{1}=\left(\frac{a(\tau_{1})}{a(t_1)}\right)^2.
\end{equation}
It remains to find $\tau_1$. To that end, define the function $\sigma(\tau)$ by,
\begin{equation}\label{lem1.1}
\sigma(\tau)\equiv\left(\frac{a(\tau)}{a(t_1)}\right)^2.
\end{equation}
From the hypotheses of the lemma, $\sigma(\tau)$ is an unbounded increasing function of $\tau$. Then by Eq.\eqref{thm2},

\begin{equation}\label{lem1.2}
\chi(\tau)\equiv\chi(\tau,\sigma(\tau))=\frac{1}{2}\int_{1}^{\sigma(\tau)}\dot{b}\left(\frac{a(\tau)}{\sqrt{\tilde{\sigma}}}\right)\frac{1}{\sqrt{\tilde{\sigma}}\sqrt{\tilde{\sigma}-1}}d\tilde{\sigma}.
\end{equation}
By hypothesis, both $a$ and $\dot{b}$ are non decreasing functions so it follows that $\chi(\tau)$ is strictly increasing. A short calculation shows that  $\chi(\tau)$ has range $[0,\infty)$ for $\tau\geq t_1$.  Thus, by continuity there must exist a unique $\tau_1\geq t_1$ such that $\chi(\tau_1)=\chi_1$.  It now follows from Eq.\eqref{sigma1} that $F(\tau_{1},\sigma_{1})=(t_1,\chi_1)$ and $F$ is a bijection.
Now consider the restricted map,
\begin{equation}
F:U\to (0,\infty)\times(0,\infty).
\end{equation}
By direct calculation, the Jacobian determinant $J(\tau,\sigma)$ is given by,
\begin{equation}\label{jacobian}
J(\tau,\sigma)=\frac{\dot{a}(\tau)}{2 \sigma}\dot{b}\left(\frac{a(\tau)}{\sqrt{\sigma}}\right)\left(\frac{\dot{b}\left(\frac{a(\tau)}{\sqrt{\sigma}}\right)}{\sqrt{\sigma-1}}+\frac{a(\tau)}{2\sqrt{\sigma}}\int_{1}^{\sigma}\frac{\ddot{b}\left(\frac{a(\tau)}{\sqrt{\tilde{\sigma}}}\right)}{\tilde{\sigma}\sqrt{\tilde{\sigma}-1}}d\tilde{\sigma}\right).
\end{equation}
Since $a,\dot{a},\dot{b}$ are positive and $\ddot{b}\geq0$, it follows immediately that $J(\sigma,\tau)>0$ on its domain, and by the inverse function theorem $F$ is a diffeomorphism.
\end{proof}

\noindent Let $G(\tau, \sigma) = (\tau, \rho(\sigma))$.  Then G is a diffeomorphism with inverse, $G^{-1}(\tau,\rho)= (\tau, \sigma(\rho))$ and non vanishing Jacobian.  Using the notation of Lemma \ref{lem1} define,

\begin{equation}\label{H}
H(t,\chi) = G\circ F^{-1}(t,\chi). 
\end{equation}
Then $H$ is a diffeomorphism from  $(0,\infty)\times(0,\infty)$ onto an open subset of $(0,\infty)\times(0,\infty)$ and may be extended to a bijection with domain  $(0,\infty)\times[0,\infty)$.  We state this result as a corollary:

\begin{corollary} \label{tau,s}
Let $a(t)$ be a smooth, increasing, unbounded function on $(0,\infty)$ with inverse function $b(t)$ satisfying $\ddot{b}(t)\geq0$ for all $t>0$. Then the function $(\tau, \rho) = H(t,\chi)$ given by Eq.\eqref{H} is a diffeomorphism from $(0,\infty)\times(0,\infty)$ onto an open subset of $(0,\infty)\times(0,\infty)$ and $H$ may be extended to a bijection with domain $(0,\infty)\times[0,\infty)$.
\end{corollary}

\begin{remark}\label{notation} Using the notation of 
Corollary \ref{affinegeodesic} we may write $Y_{\tau}(\rho)=H^{-1}(\tau, \rho)$.
\end{remark}

\begin{remark}\label{inflation}
The condition $\ddot{a}(t)\leq0$ for all $t$ for a Robertson-Walker space-time to be non inflationary is equivalent to $\ddot{b}(t)\geq0$ for all $t$.
\end{remark}

\noindent The previous corollary guarantees that $\{\tau,\rho,\theta,\varphi \}$ are the coordinates of a smooth chart for the Robertson-Walker metric given by Eq.\eqref{frwmetric}.

\begin{theorem}\label{polarform}
Let $a(t)$ be a smooth, increasing, unbounded function on $(0,\infty)$ with inverse function $b(t)$ satisfying $\ddot{b}(t)\geq0$ for all $t>0$. 
In $\{\tau,\rho,\theta,\varphi \}$ coordinates the metric of Eq.\eqref{frwmetric} is given by,

\begin{equation}\label{fermipolar}
ds^2=g_{\tau\tau} d\tau^2+d\rho^2 + \frac{a^2(\tau)}{\sigma(\rho)}S^2_k(\chi(\tau,\sigma(\rho)))d\Omega^2,
\end{equation}
where,
\begin{equation}\label{gtautau}
g_{\tau\tau}=-(\dot{a}(\tau))^{2}\left(\dot{b}\left(\frac{a(\tau)}{\sqrt{\sigma(\rho)}}\right)+a(\tau)\frac{\sqrt{\sigma(\rho)-1}}{2\sqrt{\sigma(\rho)}}\int_1^{\sigma(\rho)}\frac{\ddot{b}\left(\frac{a(\tau)}{\sqrt{\tilde{\sigma}}}\right)}{\tilde{\sigma}\sqrt{\tilde{\sigma}-1}}d\tilde{\sigma}\right)^2,
\end{equation}
and where $\sigma(\rho)$ and $\chi(\tau,\sigma(\rho))$ are given by Eqs. \eqref{inverse} and \eqref{thm2}.
\end{theorem}

\begin{proof}
By the chain rule, the derivative of $H^{-1}$ is given by $D_{H^{-1}} = D_F D_{G}^{-1}$, i.e.,
\begin{eqnarray}
 \Large\begin{pmatrix}\label{jacob2}
 \frac{\partial t}{\partial \tau} & \frac{\partial t}{\partial \rho} \\
\frac{\partial \chi}{\partial \tau} & \frac{\partial \chi}{\partial \rho}
\end{pmatrix}&
=&
\Large\begin{pmatrix}
 \frac{\partial F_1}{\partial \tau} & \frac{\partial F_1}{\partial \sigma} \\
\frac{\partial F_2}{\partial \tau} & \frac{\partial F_2}{\partial \sigma}
\end{pmatrix}
\begin{pmatrix}
 \frac{\partial G_1}{\partial \tau} & \frac{\partial G_1}{\partial \sigma} \\
\frac{\partial G_2}{\partial \tau} & \frac{\partial G_2}{\partial \sigma}
\end{pmatrix}^{-1}.
\end{eqnarray}
The second column of $D_{H^{-1}}$ is given directly by Eqs.\eqref{tsigma} and \eqref{dsigmadchi}. The entries in the first column are,
\begin{equation}\label{dtdtau}
\frac{\partial t}{\partial \tau}=\frac{\partial F_1}{\partial \tau}+D_{G^{-1}}^{2 1}\frac{\partial F_1}{\partial \sigma}
\end{equation}
and
\begin{equation}\label{dchidtau}
\frac{\partial \chi}{\partial \tau}=\frac{\partial F_2}{\partial \tau}+D_{G^{-1}}^{2 1}\frac{\partial F_2}{\partial \sigma},
\end{equation}
where $D_{G^{-1}}^{2 1}$ is the $(2,1)$ entry of $D_{G}^{-1}$. Inverting the matrix $D_{G}$ yields,
\begin{equation}\label{D21G}
D_{G^{-1}}^{2 1}=-\frac{\dot{a}(\tau)\sigma^{3/2}(\rho)\sqrt{\sigma(\rho)-1}}{a(\tau)\dot{b}\Big(\frac{a(\tau)}{\sqrt{\sigma(\rho)}}\Big)}\int_1^{\sigma(\rho)}\frac{\dot{b}\left(\frac{a(\tau)}{\sqrt{\tilde{\sigma}}}\right)\sqrt{\tilde{\sigma}}+\ddot{b}\left(\frac{a(\tau)}{\sqrt{\tilde{\sigma}}}\right)a(\tau)}{\tilde{\sigma}^2\sqrt{\tilde{\sigma}-1}}d\tilde{\sigma}.
\end{equation}
Using Eqs.\eqref{Function}, \eqref{thm1}, and \eqref{thm2}, the entries of $D_F$ may be calculated directly and are given by,
\begin{eqnarray}\label{dF1dtau}
\frac{\partial F_1}{\partial \tau}&=&\dot{b}\left(\frac{a(\tau)}{\sqrt{\sigma(\rho)}}\right)\frac{\dot{a}(\tau)}{\sqrt{\sigma(\rho)}},\\\label{dF1dsigma}
\frac{\partial F_1}{\partial \sigma}&=&-\dot{b}\left(\frac{a(\tau)}{\sqrt{\sigma(\rho)}}\right)\frac{a(\tau)}{2\sigma^{3/2}(\rho)},\\\label{dF2dtau}
\frac{\partial F_2}{\partial \tau}&=&\frac{1}{2}\int_1^{\sigma(\rho)}\ddot{b}\left(\frac{a(\tau)}{\sqrt{\tilde{\sigma}}}\right)\frac{\dot{a}(\tau)}{\tilde{\sigma}\sqrt{\tilde{\sigma}-1}}d\tilde{\sigma}
\end{eqnarray}
and
\begin{equation}\label{dF2dsigma}
\frac{\partial F_2}{\partial \sigma}=\frac{\dot{b}\left(\frac{a(\tau)}{\sqrt{\sigma(\rho)}}\right)}{2\sqrt{\sigma(\rho)}\sqrt{\sigma(\rho)-1}}.
\end{equation}
Using the definition of $\sigma$ given in Eq.\eqref{parameter}, the first two metric components of Eq.\eqref{frwmetric} may be expressed as functions of $\tau$ and $\rho$ as follows, 
\begin{eqnarray}\label{oldcomponents}
g_{tt}=-1 & \mbox{and} & g_{\chi \chi}=a(t(\tau, \rho))=\frac{a^2(\tau)}{\sigma(\rho)}.
\end{eqnarray}
In what follows, let $\{x^{0},x^{1},x^{2},x^{3}\}$ denote $\{\tau,\rho,\theta,\varphi \}$. Then using Eqs.\eqref{tsigma}, \eqref{dsigmadchi} and \eqref{oldcomponents} the coefficient $g_{\rho \rho}$ of $d\rho^{2}$ in the metric tensor is given by,
\begin{equation}
g_{\rho \rho}=g_{\alpha \beta}\frac{\partial x^{\alpha}}{\partial \rho}\frac{\partial x^{\beta}}{\partial \rho}=1,
\end{equation}
which may also be deduced by noting that the tangent vector $\partial/\partial\rho$ of the geodesic $Y_{\tau}(\rho)$ has unit length (see Corollary \ref{affinegeodesic}).
Similarly, a calculation using Eqs.\eqref{dtdtau},\eqref{D21G},\eqref{dF1dtau},\eqref{dF1dsigma} and \eqref{oldcomponents} gives,

\begin{multline}\label{long}
\indent\indent\indent\indent\indent\indent g_{\tau\tau}=g_{\alpha \beta}\frac{\partial x^{\alpha}}{\partial \tau}\frac{\partial x^{\beta}}{\partial \tau}=\frac{\dot{a}^2(\tau)}{4\sigma}\times\\
\Bigg[a^2(\tau)\left(\int_1^{\sigma}\frac{\ddot{b}\left(\frac{a(\tau)}{\sqrt{\tilde{\sigma}}}\right)}{\tilde{\sigma}\sqrt{\tilde{\sigma}-1}}d\tilde{\sigma}-\sigma\int_1^{\sigma}\frac{\ddot{b}\left(\frac{a(\tau)}{\sqrt{\tilde{\sigma}}}\right)}{\tilde{\sigma}^2\sqrt{\tilde{\sigma}-1}}d\tilde{\sigma}-\frac{\sigma}{a(\tau)}\int_1^{\sigma}\frac{\dot{b}\left(\frac{a(\tau)}{\sqrt{\tilde{\sigma}}}\right)}{\tilde{\sigma}^{3/2}\sqrt{\tilde{\sigma}-1}}d\tilde{\sigma}\right)^2\\
-\left(2\dot{b}\left(\frac{a(\tau)}{\sqrt{\tilde{\sigma}}}\right)+\sqrt{\sigma}\sqrt{\sigma-1}\int_1^{\sigma}\frac{\dot{b}\left(\frac{a(\tau)}{\sqrt{\tilde{\sigma}}}\right)\sqrt{\tilde{\sigma}}+\ddot{b}\left(\frac{a(\tau)}{\sqrt{\tilde{\sigma}}}\right)a(\tau)}{\tilde{\sigma}^2\sqrt{\tilde{\sigma}-1}}d\tilde{\sigma}\right)^2\Bigg].
\end{multline}
Applying integration by parts to the integral,
\begin{equation}
\int_1^{\sigma}\frac{\dot{b}\left(\frac{a(\tau)}{\sqrt{\tilde{\sigma}}}\right)}{\tilde{\sigma}^{3/2}\sqrt{\tilde{\sigma}-1}}d\tilde{\sigma},
\end{equation}
results in simplification of Eq.\eqref{long} and yields Eq.\eqref{gtautau}. For the off-diagonal components, a calculation using Eqs.\eqref{dchidtau}, \eqref{D21G}, \eqref{dF2dtau}, \eqref{dF2dsigma} and \eqref{oldcomponents} results in,

\begin{multline}\label{gtaurho}
g_{\tau \rho}=g_{\alpha \beta}\frac{\partial x^{\alpha}}{\partial \tau}\frac{\partial x^{\beta}}{\partial \rho}
=\dot{a}(\tau)\Bigg(\sqrt{\frac{\sigma-1}{\sigma}}\dot{b}\left(\frac{a(\tau)}{\sqrt{\tilde{\sigma}}}\right)+\frac{a(\tau)}{2}\int_1^{\sigma}\frac{\ddot{b}\left(\frac{a(\tau)}{\sqrt{\tilde{\sigma}}}\right)}{\tilde{\sigma}\sqrt{\tilde{\sigma}-1}}d\tilde{\sigma}\\
-\frac{1}{2}\int_1^{\sigma}\frac{\dot{b}\left(\frac{a(\tau)}{\sqrt{\tilde{\sigma}}}\right)\sqrt{\tilde{\sigma}}+\ddot{b}\left(\frac{a(\tau)}{\sqrt{\tilde{\sigma}}}\right)a(\tau)}{\tilde{\sigma}^2\sqrt{\tilde{\sigma}-1}}d\tilde{\sigma}\Bigg).
\end{multline}
Analogous simplification using integration by parts shows that the right hand side of Eq.\eqref{gtaurho} is identically zero.
\end{proof}

\noindent The following corollary is used in the proof of Theorem \ref{fermi}.

\begin{corollary}\label{polargeodesic}
With the same assumptions as in Theorem \ref{polarform} and using  $\{\tau,\rho,\theta,\varphi \}$ coordinates, for fixed $\theta_{0},\varphi_{0}$, and $\tau>0$, the path $Y_{\tau}(\rho)=(\tau,\rho,\theta_{0},\varphi_{0})$ is a spacelike geodesic with parameter $\rho>0$.
 \end{corollary}
 \begin{proof}
 We first express the geodesic  $Y_{\tau}(\rho)$ on $\mathcal{M}_{\theta_0,\varphi_0}$ given in Corollary \ref{affinegeodesic} in terms of $\tau, \rho$ coordinates.  Since 
\begin{equation}
 \frac{\partial}{\partial\rho}=\frac{\partial t}{\partial\rho}\frac{\partial}{\partial t}+\frac{\partial\chi}{\partial\rho}\frac{\partial}{\partial\chi},
 \end{equation}
 the tangent vector field along $Y_{\tau}(\rho)$ coincides with the tangent vector field along the path $(\tau, \rho)$ with parameter $\rho\geq0$ and $\tau$ fixed. Since these paths have the same initial data at $\gamma(\tau)$, they represent the same geodesic in the two respective coordinate systems.\\
 
\noindent  It then follows that the function $Y_{\tau}(\rho)=(\tau,\rho,\theta_{0},\varphi_{0})$ on $\mathcal{M}$ is a spacelike geodesic. That the angular coordinates are constant follows from symmetry or by solving the geodesic equations directly.
 \end{proof}

\noindent We refer to the coordinates, $\{\tau,\rho,\theta,\varphi \}$, of Theorem  \ref{polarform} as  Fermi polar coordinates.  The terminology is justified by the following theorem.

\begin{theorem}\label{fermi}
Let $a(t)$ be a smooth, increasing, unbounded function on $(0,\infty)$ with inverse function $b(t)$ satisfying $\ddot{b}(t)\geq0$ for all $t>0$. Define $x^{0}=\tau$, $x=\rho \sin\theta \cos\varphi$, $y=\rho \sin\theta \sin\varphi$, $z=\rho \cos\theta$. Then the coordinates $\{\tau,x,y,z\}$ may be extended to a chart that includes the path $\gamma(\tau)=(\tau,0,0,0)$ and $\{\partial/\partial\tau,$ $\partial/\partial x, \partial/\partial y, \partial/\partial z\}$ is a parallel tetrad along $\gamma(\tau)$.  With respect to this tetrad, $\tau,x,y,z$ are global Fermi coordinates for the observer $\gamma(\tau)$.  Expressed in these Fermi coordinates, the metric of Eq.\eqref{fermipolar} is given by,

\begin{equation}
\begin{split}\label{fermimetric}
ds^2=&\,g_{\tau\tau} d\tau^2+dx^2 +dy^2+dz^2\\ 
+&\lambda_{k}(\tau,\rho)\big[(y^2+z^2)dx^2+(x^2+z^2)dy^2+(x^2+y^2)dz^2\\
-&xy(dxdy+dydx)-xz(dxdz+dzdx)-yz(dydz+dzdy)\big],
\end{split}
\end{equation}
where $g_{\tau\tau}$ is given by Eq.\eqref{gtautau}, $\rho=\sqrt{x^2+y^2+z^2}$, and,

\begin{equation}\label{lambda}
\rho^4 \lambda_{k}(\tau,\rho) = \frac{a^2(\tau)}{\sigma(\rho)}S^2_k(\chi(\tau,\sigma(\rho)))-\rho^2.
\end{equation}
The smooth function $\lambda_{k}(\tau,\rho)$ is a function of $\tau$ and $\rho^2$, and the notation in Eq.\eqref{lambda} is the same as in Theorem \ref{general}.
\end{theorem}

\begin{proof}We begin by showing that the indicated transformation of coordinates applied to Eq.\eqref{fermipolar} results in Eq.\eqref{fermimetric}.  Using Eq.\eqref{dsigmadchi}, Eq.\eqref{fermipolar} may be rewritten as,

\begin{equation}\label{fermipolar2}
ds^2=g_{\tau\tau} d\tau^2+(d\rho^2 +\rho^{2}d\Omega^2)\\
+\left[a(\tau)\frac{S^2_k(\chi(\tau,\sigma(\rho)))}{\dot{\chi}}-\rho^{2}\right]d\Omega^2,
\end{equation}
where for convenience we take $\dot{\chi}\equiv\partial\chi/\partial\rho$. Applying the change of variables in the statement of the theorem results in,
\begin{equation}
\begin{split}\label{fermimetric2}
ds^2=&\,g_{\tau\tau} d\tau^2+dx^2 +dy^2+dz^2\\ 
+&\frac{Q_{k}(\tau,\rho)}{\rho^4}\big[(y^2+z^2)dx^2+(x^2+z^2)dy^2+(x^2+y^2)dz^2\\
-&xy(dxdy+dydx)-xz(dxdz+dzdx)-yz(dydz+dzdy)\big],
\end{split}
\end{equation}
where,

\begin{equation}\label{Q_{k}}
Q_{k}(\tau,\rho)=a(\tau)\frac{S^2_k(\chi(\tau,\sigma(\rho)))}{\dot{\chi}}-\rho^{2}.
\end{equation}
From Remark \ref{submanifold}, it follows that $\dot{\chi}$ is an even function of $\rho$, and thus $Q_{k}(\tau,\rho)$ has a smooth extension to an even function of $\rho$. Repeated use of Eqs.\eqref{tsigma} and \eqref{sigmadot} yields, $\chi(\tau,\sigma(0))=0$, $\dot{\chi}(\tau,\sigma(0))=a(\tau)^{-1}$, $\ddot{\chi}(\tau,\sigma(0))=0$, $\dddot{\chi}(\tau,\sigma(0))=2\dot{a}^2(\tau)/a^{3}(\tau)$, and  $\ddddot{\chi}(\tau,\sigma(0))=0$, where as above the overdots signify differentiation with respect to $\rho$.  Using these results, it follows by direct calculation that $Q_{k}(\tau,0)=0$ and each of the first three derivatives of $Q_{k}$ with respect to $\rho$ vanish when evaluated at $(\tau,0)$, for $k=1,0,-1$.  Thus, writing $Q_{k}(\tau,\rho)$ as a Taylor polynomial in powers of $\rho^2$, we have,

\begin{equation}\label{Q2}
Q_{k}(\tau,\rho)=\rho^{4}\lambda_{k}(\tau,\rho),
\end{equation}
where $\lambda_{k}(\tau,\rho)$ is smooth and a function of $\rho^2$, establishing Eq.\eqref{fermimetric}, which extends by continuity to the path $\gamma(\tau)=(\tau,0,0,0)$, where the metric is Minkowskian.  It now follows by calculation that all first derivatives with respect to $\tau,x,y$, or $z$ of the metric tensor components vanish on $\gamma(\tau)$, forcing the connection coefficients also to vanish on $\gamma(\tau)$. Thus, each of the vectors in the tetrad $\{\partial/\partial\tau,$ $\partial/\partial x, \partial/\partial y, \partial/\partial z\}$ is parallel along $\gamma(\tau)$.\\

\noindent Expressing the geodesic $Y_{\tau}(\rho)$ in Corollary \ref{polargeodesic} in terms of $\{\tau,x,y,z\}$ gives,

\begin{equation}\label{YinFermi}
Y_{\tau}(\rho)=(\tau,a^{1}\rho,a^{2}\rho,a^{3}\rho),
\end{equation}
where $a^{1}=\sin\theta_{0} \cos\varphi_{0}$, $a^{2}=\sin\theta_{0} \sin\varphi_{0}$, $a^{3}=\cos\theta_{0}$. The geodesic of Eq.\eqref{YinFermi} may be extended to $\gamma(t)$ and is orthogonal to $\gamma(t)$.  It now follows from Eq.\eqref{F2} that  $\tau,x,y,z$ are global Fermi coordinates for the observer $\gamma(\tau)$. 
\end{proof}

\noindent The following definition makes some of the notation in the introduction more precise and will be useful in what follows.

\begin{definition}\label{fermiobserver} 
We refer to the observer following the path $\gamma(\tau)=(\tau,0,0,0)$ given in the statement of Theorem \ref{fermi} as the \emph{Fermi observer}.  A test particle with fixed spatial  Robertson-Walker coordinates $\chi_{0},\theta_{0},\varphi_{0}$ and with world line $\gamma_{0}(\tau)=(\tau, \chi_{0}, \theta_{0},\varphi_{0})$ is said to be \emph{comoving}.  The Fermi observer is also defined to be comoving.
\end{definition}

\begin{remark}\label{referee} The unit tangent vector field, $\partial/\partial t$, of the fundamental (comoving) observers of Robertson-Walker cosmologies --- i.e., the direction of time in Robertson-Walker coordinates --- may be expressed in Fermi coordinates via,
\begin{equation}\label{d/dt} 
\frac{\partial}{\partial t}=\frac{\partial \tau}{\partial t}\frac{\partial}{\partial \tau}+\frac{\partial\rho}{\partial t}\frac{\partial}{\partial\rho}.
\end{equation}
The partial derivatives in Eq.\eqref{d/dt} are entries of the matrix $D_H(t, \chi) = D_G\circ D_{F^{-1}}(t,\chi)$ (see Eq.\eqref{H} and Corollary \ref{tau,s}) and may be found explicitly as integral expressions from Eqs.\eqref{dF1dtau} -- \eqref{dF2dsigma}, and Theorem \ref{general}.
\end{remark}

\begin{remark}\label{nonglobal}
The conclusions of Theorems \ref{polarform} and \ref{fermi} continue to hold even when the hypothesis that $\ddot{b}\geq0$ is violated, i.e., for inflationary cosmologies, but for non global charts.  The forms of the metric given by Eqs. \eqref{fermimetric} and \eqref{fermipolar} are valid in that case on some neighborhood of the Fermi observer's path $\gamma(t)$.
\end{remark}

\noindent \textbf{{\normalsize 4. Fermi relative velocities and the proper radius of $\mathcal{M}_{\tau}$}}\\

\noindent In this section we find a general bound for the finite proper radius of the Fermi space slice of $\tau$-simultaneous events, $\mathcal{M}_{\tau}$, and we obtain expressions for Fermi velocities of comoving test particles relative to the Fermi observer.  Exact results are given for the case that the scale factor has the form $a(t)=t^{\alpha}$ for some  $0< \alpha\leq1$. We begin with a definition. 

\begin{definition}\label{radiusMtau}
Define the proper radius, $\rho_{\mathcal{M}_{\tau}}$, of the Fermi space slice of $\tau$-simultaneous events, ${\mathcal{M}_{\tau}}$, by,

\begin{equation}\label{thm3'}
\rho_{\mathcal{M}_{\tau}}=\frac{a(\tau)}{2}\int_{1}^{\sigma_{\infty}(\tau)}\dot{b}\left(\frac{a(\tau)}{\sqrt{\sigma}}\right)\frac{1}{\sigma^{3/2}\sqrt{\sigma-1}}d\sigma,
\end{equation}
where $\sigma_{\infty}(\tau)$ is given by Eq.\eqref{range}.
\end{definition}

\noindent The Hubble parameter, $H$ is defined by,

\begin{equation}\label{hubbleparameter}
H = \frac{\dot{a}(\tau)}{a(\tau)}
\end{equation}
and the Hubble radius is defined to be $1/H$, i.e., the speed of light divided by $H$.

\begin{theorem}\label{radius}
Let $a(t)$ be a smooth, increasing, unbounded function on $(0,\infty)$ with inverse function $b(t)$ satisfying $\ddot{b}(t)\geq0$ for all $t>0$.  Then,
\begin{enumerate}
  \item[(a)] at proper time $\tau$ of the Fermi observer, the proper distance $\rho$ to any spacetime point along a geodesice on the space slice, $\mathcal{M}_{\tau}$, satisfies the inequality,

\begin{equation}
\rho < \rho_{\mathcal{M}_{\tau}}\leq\frac{1}{H},
\end{equation}
and $ \rho_{\mathcal{M}_{\tau}}$ is a monotone increasing function of time $\tau$.
 \item[(b)] synchronous time $t$ decreases to zero along any spacelike geodesic, $Y_{\tau}(\rho)$, orthogonal to the path of the Fermi observer at fixed proper time $\tau$, as the proper distance $\rho\rightarrow\rho_{\mathcal{M}_{\tau}}$, and $t$ is strictly decreasing as a function of $\rho$.
\end{enumerate}

\end{theorem}

\begin{proof}
By hypothesis, $\dot{b}$ is an increasing function, so from Eq.\eqref{thm3}, the proper distance from the observer to a spacetime point on $\mathcal{M}_{\tau}$ corresponding to any parameter value $\sigma$ satisfies,

\begin{equation}\label{thm3'}
\begin{split}
\rho&=\frac{a(\tau)}{2}\int_{1}^{\sigma}\dot{b}\left(\frac{a(\tau)}{\sqrt{\tilde{\sigma}}}\right)\frac{1}{\tilde{\sigma}^{3/2}\sqrt{\tilde{\sigma}-1}}d\tilde{\sigma}\\
&<\frac{a(\tau)}{2}\int_{1}^{\sigma_{\infty}(\tau)}\dot{b}\left(\frac{a(\tau)}{\sqrt{\tilde{\sigma}}}\right)\frac{1}{\tilde{\sigma}^{3/2}\sqrt{\tilde{\sigma}-1}}d\tilde{\sigma}\equiv \rho_{\mathcal{M}_{\tau}}\\
&\leq\frac{a(\tau)}{2} \,\dot{b}(a(\tau)) \int_{1}^{\infty}\frac{1}{\tilde{\sigma}^{3/2}\sqrt{\tilde{\sigma}-1}}d\tilde{\sigma}\\
&= \frac{a(\tau)}{2}\, \frac{1}{\dot{a}(\tau)}\,2 = \frac{1}{H},
\end{split}
\end{equation}
It follows immediately from Eqs.\eqref{range} and \eqref{thm3'} that $ \rho_{\mathcal{M}_{\tau}}$ increases with $\tau$. This proves part (a).  To prove (b), observe that since $a:(0,\infty)\rightarrow(a_{\text {inf}}, \infty)$, then $b:(a_{\text {inf}}, \infty)\rightarrow(0,\infty)$ and by hypothesis $b$ is an increasing function. Thus, $\lim_{a\rightarrow a_{\text {inf}}} b(a)\equiv b(a_{\text {inf}}^{+})=0$.  Using the notation of Theorem \ref{general}, Eq.\eqref{inverse}, and Corollary \ref{affinegeodesic}, we have,

\begin{equation}\label{tobigbang}
\lim_{\quad\rho\rightarrow\rho_{\mathcal{M}_{\tau}}}t(\tau,\sigma(\rho))=\lim_{\,\sigma\rightarrow\sigma_{\infty}(\tau)}t(\tau,\sigma)=\lim_{\,\sigma\rightarrow\sigma_{\infty}(\tau)}b\left(\frac{a(\tau)}{\sqrt{\sigma}}\right)=b(a_{\text {inf}}^{+})=0.
\end{equation}
It follows from Eq.\eqref{tsigma} that $dt/d\rho<0$ except at $\rho=0$, so $t$ is strictly decreasing as a function of $\rho$.

\end{proof}

\noindent The following corollary is an immediate consequence of Theorem \ref{radius}b.

\begin{corollary}\label{disjoint}
Under the hypotheses of Theorem \ref{radius}, no two distinct space-time points are simultaneous with respect to both synchronous time $t$ and Fermi time $\tau$ when $t=\tau$. 
\end{corollary}

\begin{remark}\label{referee2} The Fermi hypersurfaces of $\tau$-simultaneous events, $\{\mathcal{M}_{\tau}\}$, are isotropic with respect to the Fermi observer, but they fail to be homogeneous, in contrast to the $t$-simultaneous hypersurfaces $\{\Sigma_{t}\}$ for synchronous time $t$ (c.f. \cite{ellis}).  This is explained by Corollary \ref{disjoint} and Theorem \ref{radius}b which shows that $t$ decreases with proper distance along each $\mathcal{M}_{\tau}$, so the intrinsic curvature is not constant on $\mathcal{M}_{\tau}$.  We discuss the significance of Theorem \ref{radius} in the concluding section.
\end{remark}

\noindent It is well-known that the motion of a comoving test particle follows Hubble's Law, 
\begin{equation}\label{hubble}
v_{H}(\chi_{0})=\dot{a}(\tau)\chi_{0}=Hd.
\end{equation}
Here the Hubble speed $v_{H}(\chi_{0})\equiv\dot{d}$, the Hubble parameter $H$ is given by Eq.\eqref{hubbleparameter}, and $d=a(\tau)\chi_{0}$ is the proper distance along the spacelike path
$Z_{\tau}(\chi)=(\tau, \chi, \theta_{0},\varphi_{0})$ as $\chi$ varies from $0$ to $\chi_{0}$.  Both $Z_{\tau}(\chi)$ and the path $Y_{\tau}(\rho)$ given by Eq.\eqref{YinFermi} --- described in different coodinate systems --- are orthogonal to the path of the Fermi observer at the spacetime point $\gamma(\tau)$, but $Y_{\tau}(\rho)$ is geodesic whereas $Z_{\tau}(\chi)$ is not.\\

\noindent From Eqs. \eqref{thm2} and \eqref{thm3}, it follows the the coordinate $\chi$ is a smooth, increasing function of $\rho$ along the geodesic $Y_{\tau}(\rho)$. We denote the inverse of that function (with fixed $\tau$) by $\rho(\tau,\chi)$.  The Fermi speed, $v_{F}(\chi_{0})$, of the radially receding, comoving test particle with world line $\gamma_{0}(\tau)=(\tau, \chi_{0}, \theta_{0},\varphi_{0})$, relative to the observer $\gamma(\tau)$, is given by,

\begin{equation}\label{vf}
v_{F}(\chi_{0})=\frac{d}{d\tau}\rho(\tau,\chi_{0})\equiv \dot{\rho}.
\end{equation}
Eq.\eqref{vf} follows from Prop. 3 in \cite{bolos} and is a special case of the more generally defined Fermi relative velocity for test particles and observers following arbitrary world lines. 
\begin{remark}\label{hubblefermi}
In analogy to a well-known expression for the Hubble speed of a test particle with peculiar velocity, the following identity holds for the Fermi relative speed of a comoving particle,

\begin{equation}\label{identity}
v_{F}(\chi_{0})= H(\tau)\rho + a(\tau)\frac{d}{d\tau}\left(\frac{\rho}{a(\tau)}\right),
\end{equation}
as may be verified by direct calculation.  The second term on the right side of Eq.\eqref{identity} is roughly analogous to the peculiar velocity in Robertson-Walker coordinates.
\end{remark}

\noindent The following Theorem provides a general expression for the Fermi speed $v_{F}(\chi_{0})$.

\begin{theorem}\label{thmfs}
Let $a(t)$ be a smooth, increasing function of $t$ with inverse function $b(t)$. The Fermi speed, $v_{F}(\chi_{0})$, of the comoving test particle with world line $\gamma_{0}(\tau)$, relative to the Fermi observer, is given by,

\begin{equation}
\begin{split}\label{fs}
v_{F}(\chi_{0})=&\frac{\dot{a}(\tau)}{2}\Bigg(\int_1^{\sigma_{0}}\frac{\dot{b}\left(\frac{a(\tau)}{\sqrt{\sigma}}\right)}{\sigma^{3/2}\sqrt{\sigma-1}}d\sigma\\
+&a(\tau)\int_1^{\sigma_{0}}\frac{\ddot{b}\left(\frac{a(\tau)}{\sqrt{\sigma}}\right)}{\sigma^2\sqrt{\sigma-1}}d\sigma-\frac{a(\tau)}{\sigma_{0}}\int_1^{\sigma_{0}}\frac{\ddot{b}\left(\frac{a(\tau)}{\sqrt{\sigma}}\right)}{\sigma\sqrt{\sigma-1}}d\sigma\Bigg),
\end{split}
\end{equation}
where $\sigma_{0}$ is the unique solution to $\chi(\tau,\sigma_{0})=\chi_{0}$ in Eq.\eqref{thm2}. 
\end{theorem}
\begin{proof}
With $\chi_{0}$ fixed, differentiating Eq.\eqref{thm3} with respect to $\tau$ yields,
\begin{multline}\label{fs1}
\frac{d}{d\tau}\rho(\tau,\chi_{0})=\frac{\dot{a}(\tau)}{2}\left[\int_1^{\sigma_{0}}\frac{\dot{b}\left(\frac{a(\tau)}{\sqrt{\sigma}}\right)}{\sigma^{3/2}\sqrt{\sigma-1}}d\sigma+a(\tau)\int_1^{\sigma_{0}}\frac{\ddot{b}\left(\frac{a(\tau)}{\sqrt{\sigma}}\right)}{\sigma^2\sqrt{\sigma-1}}d\sigma\right]\\
+\frac{a(\tau)}{2}\dot{b}\left(\frac{a(\tau)}{\sqrt{\sigma_{0}}}\right)\frac{1}{\sigma_{0}^{3/2}\sqrt{\sigma_{0}-1}}\frac{d\sigma_{0}}{d\tau}.
\end{multline}
Now differentiating Eq.\eqref{thm2} with respect to $\tau$ gives,
\begin{equation}\label{fs2}
\frac{d\chi_{0}}{d\tau}=0=\frac{\dot{a}(\tau)}{2}\int_1^{\sigma_{0}}\frac{\ddot{b}\left(\frac{a(\tau)}{\sqrt{\sigma}}\right)}{\sigma\sqrt{\sigma-1}}d\sigma+\frac{1}{2}\dot{b}\left(\frac{a(\tau)}{\sqrt{\sigma_{0}}}\right)\frac{1}{\sqrt{\sigma_{0}}\sqrt{\sigma_{0}-1}}\frac{d\sigma_{0}}{d\tau}.
\end{equation}
Multiplying Eq.\eqref{fs2} by $a(\tau)/\sigma_{0}$ and solving for the last term on the right hand side of Eq.\eqref{fs1} gives the desired result.
\end{proof}

\begin{corollary}\label{fsincreasing}
Let $a(t)$ be a smooth, increasing function of t with inverse $b(t)$ such that $\ddot{b}\geq0$, then $v_{F}(\chi_{0})$ is a monotone increasing function of $\chi_{0}$.
\end{corollary}
\begin{proof}
Differentiating Eq.\eqref{fs} with respect to $\sigma_{0}$ gives,
\begin{equation}
\frac{\dot{a}(\tau)}{2\sigma_{0}^2}\left(\frac{\sqrt{\sigma_{0}}}{\sqrt{\sigma_{0}-1}}\dot{b}\left(\frac{a(\tau)}{\sqrt{\sigma_{0}}}\right)+a(\tau)\int_1^{\sigma_{0}}\frac{\ddot{b}\left(\frac{a(\tau)}{\sqrt{\sigma}}\right)}{\sigma\sqrt{\sigma-1}}d\sigma\right).
\end{equation}
With the assumption that $\ddot{b}\geq0$, we see this expression is positive on its domain. Since $\sigma_{0}$ is an increasing function of $\chi_{0}$, the result follows from the chain rule.
\end{proof}

\noindent In the next corollary we consider the class of Robertson-Walker spacetimes for which the scale factor has the form,
\begin{equation}\label{aalpha}
\begin{array}{lr}
a(t)=t^{\alpha} & 0<\alpha\leq 1.
\end{array}
\end{equation}

\noindent It is easily checked that for these models, $\ddot{b}(t)\geq0$, and therefore by Theorem \ref{fermi}, Fermi coordinates are global. This class of spacetimes includes the Milne universe ($\alpha =1$), radiation-dominated universe ($\alpha =1/2$), and matter-dominated universe ($\alpha =2/3$) considered in the next section. The following corollary is a consequence of Theorem \ref{thmfs}.

\begin{corollary}\label{alphabound}
In Robertson-Walker spacetimes with $a(t)=t^{\alpha}$ for $0<\alpha\leq 1$, the Fermi speed of the comoving test particle with world line $\gamma_{0}(\tau)$, relative to the Fermi observer is time independent and is given by,
\begin{equation}\label{fsalpha}
v_{F}(\chi_{0})=\frac{1}{2\alpha}\bigg(\int_1^{\sigma_{0}}\frac{1}{\sigma^{\frac{1}{2\alpha}+1}\sqrt{\sigma-1}}d\sigma+\frac{\alpha-1}{\sigma_{0}}\int_1^{\sigma_{0}}\frac{1}{\sigma^{\frac{1}{2\alpha}}\sqrt{\sigma-1}}d\sigma\bigg).
\end{equation}
The least upper bound for $\{v_{F}(\chi_{0})\}$ is given by,
\begin{equation}\label{upperbound}
\lim_{\chi_{0} \to\infty}v_{F}(\chi_{0})=\frac{\sqrt{\pi}\,\,\Gamma(\frac{1}{2\alpha}+\frac{1}{2})}{2\alpha\,\,\Gamma(\frac{1}{2\alpha}+1)}.
\end{equation}
The right side of Eq. \eqref{upperbound} is bounded above by $1/\alpha$ with equality only for $\alpha=1$.
\end{corollary}
\begin{proof}
By Lemma \ref{fsincreasing}, $v_{F}(\chi_{0})$ is a strictly increasing function.  It follows from Lemma \ref{lem1} that the limit on the left hand side of Eq.\eqref{upperbound} may be found by taking the limit of Eq.\eqref{fsalpha} as $\sigma_{0}\rightarrow\sigma_{\infty}(\tau)\equiv\infty$.  An application of L'H\^opital's rule shows that the limit of the second term in Eq.\eqref{fsalpha} is zero. Consider the first integral in Eq. \eqref{fsalpha}. For $0< \alpha\leq 1$,
\begin{equation}
\frac{1}{\sigma^{\frac{1}{2\alpha}+1}\sqrt{\sigma-1}}\leq \frac{1}{\sigma^{\frac{3}{2}}\sqrt{\sigma-1}},
\end{equation}
with equality only for $\alpha=1$. Thus,
\begin{equation}
\int_1^{\sigma_{0}}\frac{1}{\sigma^{\frac{1}{2\alpha}+1}\sqrt{\sigma-1}}d\sigma\leq 2\sqrt{\frac{\sigma_{0}-1}{\sigma_{0}}}.
\end{equation}
Therefore,
\begin{equation}
\frac{\sqrt{\pi}\,\,\Gamma(\frac{1}{2\alpha}+\frac{1}{2})}{\Gamma(\frac{1}{2\alpha}+1)}=\int_1^{\infty}\frac{1}{\sigma^{\frac{1}{2\alpha}+1}\sqrt{\sigma-1}}d\sigma\leq 2,
\end{equation}
from which the limit and upper bound follow.
\end{proof}

\noindent The proof of the following corollary follows by direct calculation. 

\begin{corollary}\label{geometry}
In Robertson-Walker spacetimes with $a(t)=t^{\alpha}$ for $0<\alpha\leq 1$, the Fermi speed of the comoving test particle with world line $\gamma_{0}(\tau)$ satisfies the following relationship,

\begin{equation}\label{distanceovertime}
v_{F}(\chi_{0})=\frac{\rho}{\tau}+\frac{\alpha-1}{2\alpha\sigma_{0}}\int_1^{\sigma_{0}}\frac{1}{\sigma^{\frac{1}{2\alpha}}\sqrt{\sigma-1}}d\sigma,
\end{equation}
where the proper distance $\rho$ of the test particle from the Fermi observer is given by Eq.\eqref{thm3}.
\end{corollary}

\begin{corollary}\label{geometry2}
In Robertson-Walker spacetimes with $a(t)=t^{\alpha}$ for $0<\alpha\leq 1$, the proper radius $\rho_{\mathcal{M}_{\tau}}$ of the Fermi space slice of $\tau$-simultaneous events, $\mathcal{M}_{\tau}$, is a linear function of $\tau$ and is given by,

\begin{equation}\label{Mradius}
\rho_{\mathcal{M}_{\tau}}=\tau \frac{\sqrt{\pi}\,\,\Gamma(\frac{1}{2\alpha}+\frac{1}{2})}{2\alpha\,\,\Gamma(\frac{1}{2\alpha}+1)}.
\end{equation}
\end{corollary}

\begin{proof}
The result follows by taking the limit as $\sigma_{0}\rightarrow\infty$ of of both sides of Eq.\eqref{distanceovertime}, applying Eq.\eqref{upperbound}, and observing as in the proof of Corollary \ref{alphabound}, that the second term on the right side of Eq.\eqref{distanceovertime} converges to zero.
\end{proof}

\begin{remark}\label{increasinggamma}
Since the coefficient of $\tau$ in Eq.\eqref{Mradius} is a decreasing function of $\alpha$, the proper radii, at fixed $\tau$, of the $\tau$-simultaneous events, $\mathcal{M}_{\tau}$, decrease as functions of $\alpha$ in Robertson-Walker spacetimes with $a(t)=t^{\alpha}$ for $0<\alpha\leq 1$. Conversely, $\rho_{\mathcal{M}_{\tau}}\rightarrow \infty$ for any fixed $\tau>0$ as $\alpha\rightarrow 0^{+}$.
\end{remark}

\noindent \textbf{{\normalsize 5.  Particular Cosmologies}}\\

\noindent In this section we apply results of the previous sections to particular cosmologies: the Milne universe, de Sitter universe, radiation-dominated universe, and matter-dominated universe. We include the first of these for purposes of illustration only, as the results are well known (see e.g., \cite{cook}).   In the inflationary de Sitter universe, the expressions for the metric components in Fermi coordinates for a timelike geodesic observer are also known \cite{CM} and \cite{KC3}, but we show here that co-moving particles necessarily recede from the observer only with Fermi velocities less than the speed of light, in contrast to their Hubble velocities (see Eq.\eqref{hubble}). To our knowledge, the results below for the radiation-dominated and matter-dominated universes are new.\\

\noindent \emph{The Milne Universe}\\

\noindent The Milne Universe is a special case of a Robertson-Walker spacetime and a useful prototype cosmology.  It is a solution to the field equations with no matter, radiation, or vaccum energy. For this space-time, $k=-1$ and
\begin{equation}\label{milnea}
a(t)=t,
\end{equation}
and we have for the inverse function $b$,
\begin{equation}\label{milneb}
\begin{array}{lcr}
b(t)=t & \dot{b}=1 &\ddot{b}=0\geq0.\\
\end{array}
\end{equation}
To find Fermi coordinates for a comoving observer, we first integrate Eq.\eqref{thm3} with the result,

\begin{equation}\label{milne}
\rho=\tau\sqrt{\frac{\sigma-1}{\sigma}},
\end{equation}
and thus,

\begin{equation}\label{sigmarho}
\sigma = \frac{1}{1-\left(\frac{\rho}{\tau}\right)^{2}}.
\end{equation}
It now follows from Eqs.\eqref{thm1} and \eqref{thm2} that,

\begin{equation}\label{milne1}
t=\frac{\tau}{\sqrt{\sigma}}= \sqrt{\tau^2-\rho^2},
\end{equation}
and
\begin{equation}\label{milne2}
\chi=\ln(\sqrt{\sigma}+\sqrt{\sigma-1})=\tanh^{-1}\sqrt{\frac{\sigma-1}{\sigma}}=\tanh^{-1}\left(\frac{\rho}{\tau}\right).
\end{equation}
Eqs.\eqref{milne1} and \eqref{milne2} are easily inverted to give,

\begin{equation}\label{milne3}
\begin{split}
\tau=&t\cosh\chi\\
\rho=&t\sinh\chi.
\end{split}
\end{equation}

\noindent The proper radius, $\rho_{\mathcal{M}_{\tau}}$, of the Fermi space slice of $\tau$-simultaneous events, ${\mathcal{M}_{\tau}}$ given by Eq.\eqref{Mradius} for this example is,

\begin{equation}\label{lightcone}
\rho_{\mathcal{M}_{\tau}} =\tau \frac{\sqrt{\pi}\,\,\Gamma(1)}{2\,\,\Gamma(\frac{3}{2})}=\tau=\frac{1}{H},
\end{equation}
which shows that the upper bound given by Theorem \ref{radius} is sharp. Theorem \ref{fermi} guarantees that the Fermi coordinates defined by the coordinate transformation determined by Eqs.\eqref{milne1} and \eqref{milne2} are global, and Eq.\eqref{fermipolar} for this case reduces to the polar form of the Minkowski line element,
\begin{equation}\label{minkowskimilne}
ds^2=-d\tau^2+d\rho^2+\rho^2d\Omega^2.
\end{equation}
We can thus recover from Eqs.\eqref{lightcone} and \eqref{minkowskimilne} the well known result that the Milne universe may be identified as the forward light cone in Minkowski space-time, foliated by negatively curved hyperboloids orthogonal to the time axis.   The orthogonal spacelike geodesic $Y_{\tau}(\rho)$ expressed in Fermi (i.e., Minkowski) coordinates has the form given by Eq.\eqref{YinFermi}:

\begin{equation}\label{YinFermi2}
Y_{\tau}(\rho)=(\tau,a^{1}\rho,a^{2}\rho,a^{3}\rho),
\end{equation}
where $(a^{1},a^{2},a^{3})$ is any unit vector in $\mathbb{R}^{3}$, i.e., $Y_{\tau}(\rho)$ is a horizontal line segment orthogonal to the vertical time axis in Minkowski space. Now from Eqs.\eqref{fsalpha} and \eqref{milne2}, the Fermi speed of a comoving test particle, with fixed spatial coordinate $\chi$ at proper time $\tau$, corresponding to parameter $\sigma$, is given by,
\begin{equation}\label{milne2'}
v_F=\sqrt{\frac{\sigma-1}{\sigma}}=\tanh\chi=\frac{\rho}{\tau},
\end{equation}
which by Eq.\eqref{lightcone} cannot reach or exceed the speed of light.  Comoving test particles have constant Fermi speeds proportional to their distances from the observer.\\

\noindent Although not new for this example, Fermi coordinates for a Fermi observer in the Milne universe lead to an interpretation of the Milne universe that is not immediately available via the original Robertson-Walker coordinates.  The ``big bang'' may be identified as the origin of Minkowski coordinates, and space-time itself may be defined as the set of all possible space-time points in Minkowski space that can be occupied by a test particle whose world line includes the origin of coordinates, i,e., the big bang.  Space does not expand, rather, idealized test particles from an initial ``explosion'' merely fly apart from the Fermi observer in all directions with sub light Fermi velocities. Similar interpretations were given in \cite{rindler, page, ellis}.\\

\noindent \emph{The de Sitter Universe}\\

\noindent The line element for the de Sitter Universe with Hubble's constant $H_0>0$ is given by Eq.\eqref{frwmetric} with $k=0$ and, 
\begin{equation}
\begin{array}{lcr}
a(t)=e^{H_0 t} & \dot{a}(t)=H_0 e^{H_0 t},\\
\end{array}
\end{equation}
and thus,
\begin{equation}\label{dSb}
\begin{array}{lcr}
b(t)=\frac{1}{H_0}\ln(t) & \dot{b}(t)=\frac{1}{H_0 t} & \ddot{b}(t)=-\frac{1}{H_0 t^2}<0.\\
\end{array}
\end{equation} 
The de Sitter universe is a vacuum solution to the field equations with cosmological constant $\Lambda = 3H_{0}^{2}$. Eqs.\eqref{thm1},\eqref{thm2} and \eqref{thm3} yield,

\begin{equation}\label{dS3}
\rho=\frac{1}{H_0}\sec^{-1}(\sqrt{\sigma}),
\end{equation}
and,

\begin{equation}\label{dS1}
t=\tau-\frac{1}{H_0}\ln{\sqrt{\sigma}}=\tau+\frac{\ln(\cos(H_0\rho))}{H_0},
\end{equation}
\begin{equation}\label{dS2}
\chi=\frac{\sqrt{\sigma-1}}{H_0 e^{H_0\tau}}=\frac{e^{-H_0\tau}\tan(H_0\rho)}{H_0}.
\end{equation}
It follows immediately from Eq.\eqref{dS3} that $H_0\rho <\pi/2$ (note that the hypothesis to Theorem \ref{radius} is violated here). If synchronous time $t$ is required to be positive, then it follows from Eq.\eqref{dS3} together with Eqs.\eqref{parameter} and \eqref{range} that $\sqrt{\sigma}<\exp(H_{0}\tau)$, and therefore,

\begin{equation}
e^{-H_0\tau}<\cos(H_0 \rho).
\end{equation} 
Thus, along the spacelike geodesics,

\begin{equation}
\chi<\frac{\sin(H_0 \rho)}{H_0}<\frac{1}{H_0}.
\end{equation}
Although Fermi coordinates are not global for this example, in light of Remark \ref{nonglobal} we may calculate the metric coefficients in $\{\tau,\rho,\theta,\varphi\}$ coordinates by inserting Eqs.\eqref{dSb}, \eqref{dS3}, and \eqref{dS2} into Eqs.\eqref{fermipolar} and \eqref{gtautau}.  The result is,
\begin{equation}
ds^2=-\cos^2(H_0\rho)d\tau^2+d\rho^2+\frac{\sin^2(H_{0}\rho)}{H_{0}^2}d\Omega^{2},
\end{equation}
which is the same expression obtained in \cite{CM}. A description of the way in which Fermi coordinates break down at the boundary of the Fermi chart in de Sitter space was included in \cite{KC3}.  It is intriguing to observe that maximal Fermi charts in both the Milne and de Sitter universes each occupy a single ``quadrant'' of larger embedding space-times, \cite{hawking}. \\

\noindent The Fermi relative speed of a comoving test particle at time $\tau$ with fixed space coordinate $\chi$, corresponding to the parameter $\sigma$, is given by Theorem \ref{thmfs} and reduces to,
\begin{equation}
v_F=\frac{\sqrt{\sigma-1}}{\sigma}.
\end{equation}
Combining this with Eq.\eqref{dS2} gives,
\begin{equation}
v_F(\chi)=\frac{H_{0}e^{H_{0}\tau}\chi}{1+(H_{0}e^{H_{0}\tau}\chi)^2}.
\end{equation}
Thus, the Fermi relative speed of a comoving test particle is bounded by one-half the speed of light for all values
of $H_0$, $\chi$ and $\tau$. Although Hubble and Fermi speeds are not directly comparable, as we discuss in the concluding section, the above expression is strikingly different from the standard formula for the Hubble speed of a comoving test particle,
\begin{equation}\label{hubspeeddS}
v_{H}(\chi)=H_{0}e^{H_{0}\tau}\chi,
\end{equation} 
which is unbounded.\\

\noindent \emph{Radiation-Dominated Universe}\\

\noindent The radiation-dominated universe is characterized by $k=0$ and the scale factor, 
\begin{equation}\label{rada}
a(t)=\sqrt{t},
\end{equation}
and thus,
\begin{equation}\label{radb}
\begin{array}{lcr}
b(t)=t^2 & \dot{b}=2t &\ddot{b}=2\geq0.\\
\end{array}
\end{equation}
Eqs.\eqref{thm1}, \eqref{thm2} and \eqref{thm3} yield,
\begin{equation}\label{rad1}
t=\frac{\tau}{\sigma},
\end{equation}
\begin{equation}\label{rad2}
\chi=2\sqrt{\tau}\sec^{-1}\sqrt{\sigma}, 
\end{equation}
and
\begin{equation}\label{rad}
\rho=\tau\left(\frac{\sqrt{\sigma-1}}{\sigma}+\sec^{-1}\sqrt{\sigma}\right).
\end{equation}
By Theorem \ref{radius} the proper radius of the space slice of $\tau$-simultaneous events, $\mathcal{M}_{\tau}$, is bounded by $2\tau$. The exact value, given by Eq.\eqref{Mradius}, is,

\begin{equation}\label{lightcone2}
\rho_{\mathcal{M}_{\tau}} =\tau \frac{\sqrt{\pi}\,\,\Gamma(\frac{3}{2})}{\Gamma(2)}=\frac{\pi}{2}\tau.
\end{equation}

\noindent By Corollary \ref{tau,s}, $t$ and $\chi$ are smooth functions of $\tau$ and $\rho$. As in Eq.\eqref{inverse} we write $\sigma=\sigma_{\tau}(\rho)$. Then from Eq.\eqref{fermipolar}, the line element for the radiation-dominated universe in polar Fermi coordinates is,
\begin{equation}
ds^2=-\frac{1}{\sigma}\left(1+\sqrt{\sigma-1}\,\sec^{-1}\sqrt{\sigma}\right)^2d\tau^2+d\rho^2+\frac{1}{\sigma}\left(2\tau\sec^{-1}\sqrt{\sigma}\right)^2d\Omega^2.
\end{equation}
The Fermi relative speed of a comoving test particle with fixed coordinate $\chi$ corresponding to parameter $\sigma$ may be calculated from Eq.\eqref{fsalpha} as,
\begin{equation}\label{fermivelrad}
v_{F}=\frac{\sqrt{\sigma-1}}{\sigma}+\frac{\sigma-1}{\sigma}\sec^{-1}\sqrt{\sigma}.
\end{equation}
Using \eqref{rad2}, $v_F$ can also be expressed in terms of $\chi$ as,
\begin{equation}
v_{F}=\frac{1}{2}\sin\Big(\frac{\chi}{\sqrt{\tau}}\Big)+\frac{\chi}{4\sqrt{\tau}}\Big(1-\cos\Big(\frac{\chi}{\sqrt{\tau}}\Big)\Big).
\end{equation}
Applying Corollary \ref{alphabound}, we find that the asymptotic limit of the Fermi relative speed of a comoving test particle is $\pi/2$ times the speed of light.\\  

\noindent It follows from Corollary \ref{geometry} or by directly comparing Eqs.\eqref{rad} and \eqref{fermivelrad} that,

\begin{equation}
v_{F} =\frac{\rho}{\tau}-\frac{\sec^{-1}\sqrt{\sigma}}{\sigma},
\end{equation}
so that for large $\sigma$, or equivalently for large proper distance $\rho$, $v_{F} \approx\rho/\tau$, in analogy to the Milne (or Minkowski) universe.  However, in the Milne universe, the proper distance from the Fermi observer at time $\tau$ is bounded by, and asymptotically equal to, $\tau$. Thus, $v_{F}= \rho/\tau<1$. By contrast, the corresponding bound in the radiation-dominated universe, i.e., the radius of $\mathcal{M}_{\tau}$, is $(\pi/2)\tau$ so that $v_{F}=\rho/\tau<\pi/2$, with asymptotic equality. From this point of view, the existence of superluminal Fermi velocities in the radiation-dominated universe may be attributed to the greater diameters of the Fermi space slices, $\{\mathcal{M}_{\tau}\}$, in comparison to the Milne space-time.\\

\noindent \emph{Matter-Dominated Universe}\\

\noindent The final case we consider is the matter-dominated universe. For this 
spacetime  $k=0$ and the scale factor is given by
\begin{equation}\label{mata}
a(t)=t^{\alpha},
\end{equation}
where $\alpha=2/3$.  The inverse of $a$ and its derivatives are given by,
\begin{equation}
\begin{array}{lcr}\label{matb}
b(t)=t^{3/2} & \dot{b}(t)=\frac{3}{2}t^{1/2} & \ddot{b}(t)=\frac{3}{4}t^{-1/2}\geq0.\\
\end{array}
\end{equation}
By Theorem \ref{general},
\begin{equation}
t=\frac{\tau}{\sigma^{3/4}},
\end{equation}
\begin{equation}\label{mat2}
\chi=3\tau^{1/3}\left[\frac{\sqrt{\pi}\,\Gamma(\frac{5}{4})}{\Gamma(\frac{3}{4})}-\frac{_2F_1(\frac{1}{4},\frac{1}{2};\frac{5}{4};\frac{1}{\sigma})}{\sigma^{1/4}}\right],
\end{equation}
and
\begin{equation}\label{mat3}
\rho=\tau\left[\frac{\sqrt{\pi}\,\Gamma(\frac{5}{4})}{\Gamma(\frac{3}{4})}-\left(\frac{\sigma-1}{\sqrt{\sigma}}\right)^{3/2}+ \,\,_2F_1\left(-\frac{3}{4},\frac{1}{2};\frac{1}{4};\frac{1}{\sigma}\right)\sigma^{3/4}\right],
\end{equation}
where $_2F_1(\cdot,\cdot;\cdot;\cdot)$ is the Gauss hypergeometric function. It follows from Theorem \ref{fermi} that the metic in Fermi polar coordinates is given by,
\begin{equation}
ds^2=g_{\tau \tau}d\tau^2+d\rho^2+\frac{\tau^{4/3}}{\sigma}\chi(\tau,\sigma)^2 d\Omega^2,
\end{equation} 
where $\sigma=\sigma_{\tau}(\rho)$, as in Eq.\eqref{inverse}, is the inverse of the function given by Eq.\eqref{mat3}, $\chi(\tau,\sigma)$ is given by Eq.\eqref{mat2}, and,

\begin{equation}
g_{\tau \tau}=-\frac{\tau^{4/3}}{\sqrt{\sigma}}\left[1+\frac{\sqrt{\sigma-1}}{\sigma^{1/4}}\left(\frac{\sqrt{\pi}\,\Gamma(\frac{5}{4})}{\Gamma(\frac{3}{4})}-\frac{_2F_1(\frac{1}{4},\frac{1}{2};\frac{5}{4};\frac{1}{\sigma})}{\sigma^{1/4}}\right)\right]^2.
\end{equation} 
From Eq.\eqref{fsalpha}, the Fermi speed of a comoving test paricle with worldline $\gamma_0$ is given by,
\begin{equation}
v_F(\chi_{0})=\frac{3}{4}\bigg(\int_1^{\sigma_{0}}\frac{1}{\sigma^{\frac{7}{4}}\sqrt{\sigma-1}}d\sigma-\frac{1}{3\sigma_{0}}\int_1^{\sigma_{0}}\frac{1}{\sigma^{\frac{3}{4}}\sqrt{\sigma-1}}d\sigma\bigg).
\end{equation}
Corollaries \ref{alphabound} and \ref{geometry2} then give the supremum of this speed and the proper radius of $\mathcal{M}_{\tau}$ as,
\begin{equation}
\lim_{\sigma_{0}\to\infty}v(\sigma_{0})=\frac{3}{4}\frac{\sqrt{\pi}\,\Gamma(\frac{5}{4})}{\Gamma(\frac{7}{4})}=\frac{\rho_{\mathcal{M}_{\tau}}}{\tau}\approx 1.31103.
\end{equation}
This shows that the matter-dominated Universe supports superluminal Fermi velocities at proper distances away from the Fermi observer, sufficiently close to the radius, $\rho_{\mathcal{M}_{\tau}}$, of the Fermi space slice of $\tau$-simultaneous events.\\

\noindent \textbf{{\normalsize 6.  Conclusions}}\\

\noindent Theorems \ref{polarform} and \ref{fermi} of this paper give transformation formulas for Fermi coordinates for observers comoving with the Hubble flow in expanding Robertson-Walker spacetimes, along with exact expressions for the metric tensors in those coordinates.  We have shown that Fermi coordinates are global for non inflationary cosmologies, i.e., when the scale factor $a$ satisfies the condition, $\ddot{a}(t)\leq0$. Global Fermi coordinates may be useful for the purpose of studying the influence of global expansion on local dynamics and kinematics \cite{CG}.\\

\noindent Our results also apply to cosmologies that include inflationary periods, though in such cases the Fermi charts are local.  However, if the space-time includes an early inflationary period, but $\ddot{a}(t)\leq0$ for all $t\geq t_{0}$, for some $t_{0}$, then by recalibrating the scale factor to $\tilde{a}(t)=a(t+t_{0})$, a global Fermi coordinate chart and all of our results are immediately available for the submanifold of space-time events with $t>t_{0}$.\\

\noindent  In Sect. 4 we found exact expressions for the Fermi relative velocities of comoving (and necessarily receding) test particles. It was shown that superluminal relative Fermi velocities exist, and that those velocities increase with proper distance from the observer.  Superluminal least upper bounds were given for cosmologies whose scale factors follow power laws. We note that although the overall qualitative behavior of the relative speeds $v_{H}$  and $v_{F}$ may be compared, it follows from Corollary \ref{disjoint} that at any given proper time of the Fermi observer, $v_{H}$ and $v_{F}$ measure speeds of the same comoving test particle (with fixed coordinate $\chi_{0}$) only when it is at different spacetime points.  For a comoving test particle at a given spacetime point, $v_{H}$ and $v_{F}$ give the particle's relative speeds at different times of the Fermi observer.  \\

\noindent The existence of superluminal relative velocities bears on the question of whether space is expanding, c.f. \cite{gron, CG, cook, confusion, KC10} and the numerous references in those papers.  On this matter, our results may be contrasted with arguments given in \cite{cook}. In that paper, the coordinate transformation for the Milne universe, repeated in our Eqs.\eqref{milne2} and \eqref{milne3}, was used to compare the Hubble and Fermi (or Minkowski) relative speeds $v_{H}$  and $v_{F}$ of comoving test particles. The incongruity of superluminal Hubble speeds and necessarily subluminal Minkowski speeds in the Milne universe was discussed. It was argued that the analogous qualitative difference would also occur for cosmologies that include matter or radiation, through a comparison of Hubble speeds to speeds defined via coordinates with a ``rigid'' radial coordinate.\footnote{c.f. Sect. V of \cite{cook}, p. 63.}  However, if the latter class of coordinates includes Fermi coordinates --- the coordinates used to deduce that conclusion for the Milne universe --- our results in Sect. 5 for the radiation-dominated and matter-dominated cosmological models do not support that conjecture.\\

\noindent For cosmological models with a scale factor of the form $a(t)=t^{\alpha}$ for $0< \alpha\leq1$, the existence of superluminal relative velocities of comoving particles may be understood in terms of the geometry of the simultaneous space slices, $\{\mathcal{M}_{\tau}\}$. In Sect. 5, it was shown through the use of specific examples that superluminal relative Fermi velocities exist provided ``there is enough space'' in the sense that the proper radius $\rho_{\mathcal{M}_{\tau}}$ of $\mathcal{M}_{\tau}$ satisfies the condition $\rho_{\mathcal{M}_{\tau}}>\tau$. 
\\ 

\noindent In Sect. 5 we also showed that the relative Fermi speeds of comoving particles are bounded by one-half the speed of light in the de Sitter universe, a space-time considered to be ``expanding.''  By way of contrast, superluminal relative Fermi velocities were proven to exist in \cite{KC10} in the static Schwarzschild space-time (with interior and exterior metric joined at the boundary of the interior fluid), a space-time that is not usually regarded as ``expanding.''  Thus, it may be argued that existence of superluminal relative velocities, in general, is not the appropriate criterion for the purpose of defining what is meant by the expansion of space.\\

\noindent  Does space expand?  An affirmative answer may be given for the non inflationary Robertson-Walker cosmologies studied in this paper, in the following sense. For any comoving geodesic observer, the Fermi space slices of $\tau$-simultaneous events, $\{\mathcal{M}_{\tau}\}$, that foliate the space-time have finite proper diameters that are increasing functions of the observer's proper time. This is Theorem \ref{radius}a. Theorem \ref{radius}b explains how it is possible for the space slices to have only finite extent. What stops these hypersurfaces of constant Fermi time from continuing beyond their proper diameters? The theorem makes precise the way in which all space-time events are simultaneous to synchronous time $t=0$, the big bang in cosmological models admitting an initial singularity. \\

\noindent \textbf{Acknowledgment.} E. Randles was partially supported during the course of this research by the Louis Stokes Alliance for Minority Participation (LSAMP) program at California State University Northridge. We thank Peter Collas for helpful suggestions.

\end{document}